\documentclass[copyright,creativecommons]{eptcs}
\providecommand{\event}{QPL 2013}

\usepackage{graphicx}
\usepackage{amsmath,amssymb,amsfonts,latexsym}

\usepackage{QED}
\usepackage{epsfig}
\usepackage{euscript}
\usepackage{tikz,ifdraft}
\usetikzlibrary{arrows,positioning,matrix,backgrounds,calc,fit}
\usepackage{url}

\newtheorem{theorem}{Theorem}[section]
\newtheorem{lemma}[theorem]{Lemma}

\newtheorem{proposition}[theorem]{Proposition}

\newcommand{\plusminus}{\pm}

\newcommand{\beqa}{\begin{eqnarray*}}
\newcommand{\eeqa}{\end{eqnarray*}\par\noindent}

\newcommand{\VV}{\mathcal{V}}

\newcommand{\WW}{\mathcal{W}}
\newcommand{\lrarr}{\longrightarrow}
\newcommand{\rarr}{\rightarrow}

\newcommand{\ie}{\textit{i.e.}~}

\newcommand{\DD}{\mathcal{D}}
\newcommand{\EE}{\mathcal{E}}

\newcommand{\Real}{\mathbb{R}}

\newcommand{\IFF}{\; \Longleftrightarrow \;}
\newcommand{\AND}{\; \wedge \;}

\newcommand{\card}[1]{|#1|}

\newcommand{\UU}{\mathcal{U}}

\newcommand{\hL}{h_{\Lambda}}
\newcommand{\hl}{h^{\lambda}}

\newcommand{\vv}{\mathbf{v}}

\newcommand{\EEj}[1]{\EE^{(#1)}}
\newcommand{\ej}[1]{e^{(#1)}}

\newcommand{\tus}{\omega_{U,s}}
\newcommand{\tusp}{\omega_{U',s'}}
\newcommand{\tusq}{\omega_{U'',s''}}
\newcommand{\vus}{\mathbf{v}[U,s]}

\newcommand{\vusq}{\mathbf{v}[U'',s'']}
\newcommand{\mus}{\mu_{U,s}}
\newcommand{\musp}{\mu_{U',s'}}
\newcommand{\musq}{\mu_{U'',s''}}

\newcommand{\sm}[1]{s_{m := #1}}

\newcommand{\ww}{\mathbf{w}}

\newcommand{\mc}{(X, \UU)}
\newcommand{\Omg}{\Omega}
\newcommand{\DS}{\DD^{\plusminus}}
\newcommand{\Prob}{\DD^{+}}
\newcommand{\mO}{m_{\Omg}}

\title{No-Signalling Is Equivalent To Free Choice of Measurements}
\author{Samson Abramsky
\institute{Department of Computer Science\\ University of Oxford}
\email{samson.abramsky@cs.ox.ac.uk}
\and Adam Brandenburger
\institute{Stern School of Business and Polytechnic School of Engineering \\ New York University}
\email{adam.brandenburger@stern.nyu.edu}
\and Andrei Savochkin
\institute{New Economic School \\ Skolkovo, Moscow 143025}
\email{asavochkin@nes.ru}
}

\def\titlerunning{No-Signalling Is Equivalent To Free Choice of Measurements}
\def\authorrunning{Samson Abramsky, Adam Brandenburger, and Andrei Savochkin}

\begin{document}

\maketitle

\begin{abstract}
No-Signalling is a fundamental constraint on the probabilistic predictions made by physical theories. It is usually justified in terms of the constraints imposed by special relativity.
However, this justification is not as clear-cut as is usually supposed.
We shall give a different perspective on this condition by showing an equivalence between No-Signalling and Lambda Independence, or ``free choice of measurements'', a condition on hidden-variable theories which is needed to make no-go theorems such as Bell's theorem non-trivial.
More precisely, we shall show that a probability table describing measurement outcomes is No-Signalling if and only if it can be realized by a Lambda-Independent hidden-variable theory of a particular canonical form, in which the hidden variables correspond to non-contextual deterministic predictions of measurement outcomes. The key proviso which avoids contradiction with Bell's theorem is that we consider hidden-variable theories with \emph{signed} probability measures over the hidden variables --- \ie negative probabilities. Negative probabilities have often been discussed in the literature on quantum mechanics. We use a result proved previously in ``The Sheaf-theoretic Structure of Locality and Contextuality'' by Abramsky and Brandenburger, which shows that they give rise to, and indeed characterize, the entire class of No-Signalling behaviours.
In the present paper, we put this result in a broader context, which reveals the surprising consequence that the No-Signalling condition is equivalent to the apparently completely different notion of free choice of measurements.
\end{abstract}

\section{Introduction}

The purpose of this note is to show a precise, surprising and heretofore unrealized connection between two important notions in the foundations of quantum mechanics.

The first of these is \emph{No-Signalling}: a condition on the probability tables arising from experimental scenarios which is usually taken to reflect the physical constraints imposed by special relativity on correlations between spatially distributed observations. This condition says that, when Alice and Bob perform measurements in spacelike-separated locations, the marginal probabilities for Alice's observations of outcomes for her measurements are independent of Bob's choice of measurement setting \cite{popescu1994quantum}. This can be taken to reflect the inability, under relativistic constraints, for information about Bob's settings to reach Alice's site in time to influence her outcomes. While this seems intuitively plausible, as far as we are aware there is no \emph{formal} derivation of this principle from precise assumptions. Also, there are subtleties lurking; notably, the relativistic ideas apply to individual histories of  space-time events (selections of measurement settings and observations of outcomes), while the probabilities refer, presumably, to repeated runs. Prima facie, it seems quite possible that there could be statistical correlations over a repeated run of measurements of the kind prohibited by No-Signalling, without requiring any super-luminal signalling.
One may try to finesse this point by taking a Bayesian view of the probabilities, but this does not reflect the usual physical interpretation of actual performance of measurements and observation of outcomes.

The second notion is \emph{Lambda-Independence} \cite{dickson1999quantum}, which is a condition on \emph{hidden-variable theories} which may be used to explain the results of observations in local realistic terms. The idea is that the empirically observed probabilities of measurement outcomes can be obtained by averaging over possible values of the hidden variable, which is taken to contain a complete description of (all relevant aspects of) physical reality. To avoid trivialising the task of recovering empirical content through hidden variables, it is necessary to make the assumption that the probability distribution on the hidden variable is independent of the probabilities on measurement settings induced by a given value of the hidden variable. Otherwise, if for example the hidden variable actually \emph{determines} which measurement settings can be chosen, then the hidden-variable theory becomes a self-fulfilling prophecy, and \emph{any} empirical behaviour can trivially be reproduced.
An alternative and more evocative name for the Lambda-Independence condition is \emph{free choice of measurements}, as captured in the probabilistic independence assumption we have described.

We note that, beyond the lack of a formal derivation of no-signalling to which we have already alluded, there are additional grounds for doubting that this condition is essentially relativistic in nature. In particular, it is satisfied by ordinary quantum mechanics, with a classical background. Indeed, it can be seen to arise purely as a property of families of commuting sets of observables, without any consideration of tensor product structure or any other reflection of space-like separation. See \cite[Section 9]{abramsky2011unified} for a careful derivation and discussion of a general form of No-Signalling in this sense.

What we shall show is that these two, apparently very different conditions are actually \emph{equivalent} in the following sense: an empirical model --- essentially a probability table for outcomes of measurements as discussed above --- is No-Signalling if and only if it can be realized by a Lambda-Independent local hidden-variable model. This result comes, of course, with an important proviso, since otherwise it would be in conflict with standard no-go results such as Bell's theorem \cite{bell1964einstein}. The proviso is that we must allow a general form of local hidden-variable model which admits \emph{signed probability measures} on the hidden variables; or, more colloquially, negative probabilities. Negative probabilities have been widely discussed in the literature on quantum mechanics; see e.g. \cite{wigner1932quantum,Dirac42,moyal1949quantum,feynman1987negative,sudarshan1993new}. We shall not discuss physical or operational interpretations of negative probabilities here; for this, see \cite{negprob}. The salient point is that negative probabilities, when used with an otherwise standard, or even restricted notion of local hidden-variable model, not only account for all behaviour which can be realized in quantum mechanics, but in fact yield all no-signalling behaviour. The precise correspondence between these notions offers a new perspective on the no-signalling condition, and demands to be understood. 

We should also mention that the setting in which we prove this result is rather general, encompassing not only the familiar type of Bell scenarios, for any number of parties, measurement settings and outcomes, but a much wider range of situations, including e.g. Kochen-Specker configurations \cite{kochen1975problem}. This follows the general setting for a unified account of non-locality and contextually developed in \cite{abramsky2011unified}. In fact, the main mathematical content of our results is essentially implicit in \cite[Section 5]{abramsky2011unified}. In the present note, we develop the conceptual significance of these ideas, and in particular we expose the equivalence between No-Signalling and Lambda-Independence.

\section{The Main Result}
\subsection{Notation}
Given sets $X$, $Y$, we write $Y^X$ for the set of functions from $X$ to $Y$. Given a function $s : X \rarr Y$ and $U \subseteq X$, we write $s |_{U} : U \rarr Y$ for the restriction of $s$ to $U$. We write $\card{X}$ for the cardinality of a finite set $X$.

A \emph{measurement cover} is a pair $\mc$, where $X$ is a finite set, and $\UU$ is a family of subsets of $X$. We generally assume that $\UU$ is an anti-chain, \ie distinct elements of $\UU$ are incomparable under set inclusion. The idea is that $X$ is a set of measurement labels, and the sets $U \in \UU$ are maximal sets of compatible measurements, which can be performed jointly.

We fix a set $O$ of \emph{measurement outcomes}, and define $\Omg := O^X$, the set of assignments of outcomes to all measurements simultaneously. We can think of $\Omg$ as  a set of \emph{canonical hidden variables}.

Given a set $S$, we write $\DS(S)$ for the set of (finite-support) \emph{signed probability measures} on $S$, \ie the set of maps $d : S \rarr \Real$ of finite support, and such that
\[ \sum_{x \in S} d(x) = 1 . \]
We write $\Prob(S)$ for the subset of $\DS(S)$ of measures valued in the non-negative reals; these are just the  probability distributions on $S$ with finite support.

Given a measurement cover $\mc$ and set of outcomes $O$, we define a set of \emph{atomic events} $E$:
\[ E := \{ (U, s) \mid U \in \UU \AND s \in O^U \} . \]
Thus $(U, s)$ is the event that the measurements in $U$ were performed, and the outcome $s(x)$ was observed for each $x \in U$.

An \emph{empirical model} over $(\mc, O)$ is a probability distribution $e \in \Prob(E)$ such that, for each $U \in \UU$,
\[ e(U) := \sum_{s \in O^U} e(U, s) \; > \; 0. \]
For each $U \in \UU$, $e$ determines a probability distribution
$e_U \in \Prob(O^U)$ as the conditional probability $e_U(s) := e((U, s) | U)$.

A \emph{signed canonical hidden-variable model} (schv model) is a signed measure $m \in \DS(\Omega \times \UU)$. A \emph{probabilistic canonical hidden-variable model} (pchv model) is an schv model $p$ such that $p \in \Prob(\Omega \times \UU)$.

Given $(U, s ) \in E$, we define
\[ \Omg(U, s) := \{ \omega \in \Omg \mid \omega |_U = s  \} . \]
This is the set of canonical hidden variables which are consistent with the atomic event $(U, s)$. 
Note that, for each $U \in \UU$, the sets  $\Omg(U, s)$ as $s$ ranges over $O^U$ partition $\Omega$.

An schv model $m$ determines a signed measure $\hat{m} \in \DS(E)$ by marginalization:
\[ \hat{m} (U,s) := \sum_{\omega \in \Omg(U, s)} m(\omega, U) . \]
We say than an schv model $m$ \emph{realizes} an empirical model $e$ if $\hat{m} = e$.

\begin{proposition}
\label{pchvprop}
For every empirical model $e$, there is a pchv model $p$ which realizes $e$.
\end{proposition}
\begin{proof}
Given $e$, we can define e.g.
\[ p(\omega, U) := \frac{e(U, s)}{| \Omg(U, s) |}, \qquad s = \omega | U . \]
Clearly $\hat{p} = e$. Moreover, 
\[ \begin{array}{lcl}
\sum_{\omega \in \Omg, U \in \UU} p(\omega, U) & = & \sum_{U \in \UU} \sum_{s \in O^U} \sum_{\omega' \in \Omg(U, s)} p(\omega', U) \\
& = & \sum_{U \in \UU} \sum_{s \in O^U} e(U, s) \; = \; 1 . 
\end{array}
\]
Of course, $p$ is far from unique: we can divide the weight $e(U, s)$ arbitrarily among the $p(\omega, U)$ for $\omega \in \Omg(U, s)$.
\end{proof}
This result shows that, without additional constraints, realization by deterministic hidden variables is trivially achieved.

The key condition is \emph{Lambda-Independence}: the distribution on the hidden variables should be statistically independent of the choice of measurement context. This can be understood as a form of free choice of measurements by the observers, unconstrained by the hidden variable. It corresponds to a standard statistical assumption, that the observed outcomes of experiments are not subject to \emph{selection bias}.

Formally, we say that an schv model $m$ satisfies Lambda-Independence if it factors as a product: $m = m_{\Omg} m_{\UU}$, where $m_{\Omg} \in \DS(\Omg)$ and $m_{\UU} \in \DS(\UU)$ are the marginals of $m$.
\begin{proposition}
An schv model $m$ is Lambda Independent if and only if the following condition holds: for all $\omega \in \Omg$, $U, U' \in \UU$,
\[ m(\omega | U) = m(\omega | U') . \]
\end{proposition}

We say that an empirical model which is realized by a pchv model satisfying Lambda Independence admits \emph{local hidden variables}. This definition is equivalent  to the standard definitions, which allow a broader class of hidden variable models. For a proof of this at the level of generality of our discussion here, see Theorem~7.1 of \cite{abramsky2011unified}.

In this language, Bell's Theorem can be stated as follows.
\begin{theorem}[Bell's Theorem]
There are empirical models which can be realized in quantum mechanics, and which do not admit local hidden variables.
\end{theorem}

For a detailed discussion of what it means for empirical models to be realized in quantum mechanics, see \cite{abramsky2011unified}.

\begin{proposition}
\label{liequivprop}
Let $m = m_{\UU}m_{\Omg}$ be a Lambda-Independent schv model. An empirical model $e$ is realized by $m$ if and only if the following conditions hold:
\[ \begin{array}{ll}
(1) & \forall U \in \UU. \; e(U) = m_{\UU}(U) \\
(2) & \forall (U, s) \in E. \; e_{U}(s) = \sum_{\omega \in \Omg(U, s)} m_{\Omg}(\omega) .
\end{array}
\]
\end{proposition}

We now consider the question of which empirical models can be realized by schv models satisfying Lambda Independence.

We say that an empirical model $e$ satisfies \emph{No-Signalling} if the following condition holds: for all $U, U' \in \UU$,
\[ e_U(s |_{U \cap U'}) = e_{U'}(s |_{U \cap U'}). \]
This says that the distributions conditioned on different measurement choices have common marginals; thus the choice of additional measurements $U \setminus V$ outside a compatible set $V$ has no effect on the observed statistics for the measurement outcomes in $V$. This is easily seen to be equivalent to the standard formulation of No-Signalling in Bell-type scenarios, and to be satisfied generally in quantum mechanics \cite{abramsky2011unified}.

\begin{proposition}
\label{liimpnsprop}
Let $m$ be an schv model satisfying Lambda Independence which realizes an empirical model $e = \hat{m}$. Then $e$ satisfies No-Signalling.
\end{proposition}
\begin{proof}
For  $U, U' \in \UU$ and $s \in O^{U \cap U'}$, let $S := \{ s' \in O^U \mid s' |_{U \cap U'} = s \}$. 
Note that, for all $\omega \in \Omg$:
\[ \omega |_U \in S \IFF \omega |_{U \cap U'} = s . \]
Now, using Proposition~\ref{liequivprop},
\[ \begin{split}
e_{U}(s) & = \; \sum_{s' \in S} e_U(s') \; = \; \sum_{s' \in S} \sum_{\omega \in \Omg(U, s')} m_{\Omg}(\omega) \\ 
& = \; \sum_{\omega \in \Omg(U \cap U', s)} m_{\Omg}(\omega) \; = \; e_{U'}(s) .
\end{split} \]
\end{proof}

Note however that the converse does not hold, even for pchv models.
\begin{proposition}
There are pchv models $p$ such that $p$ does not satisfy Lambda-Independence, while the realized empirical model $\hat{p}$ satisfies No-Signalling.
\end{proposition}
\begin{proof}
By Bell's theorem, there are quantum models $q$ which are not realized by any pchv model satisfying Lambda-Independence. By the No-Signalling theorem, any such $q$ satisfies No-Signalling. 
By Proposition~\ref{pchvprop}, $q$ is realized by a pchv model $p$; it follows that $p$ cannot satisfy Lambda-Independence.
\end{proof}

We now embark on the proof of our main result, the equivalence of No-Signalling and Lambda-Independence.

Let $d := \card{E}$ and $p := \card{\Omega}$. An empirical model induces a vector $\vv$ in the real vector space $\Real^d$, with $\vv_{(U, s)} := e_U(s)$.
Similarly, we can regard a signed distribution $\mO \in \DS(\Omg)$ as a vector in $\Real^p$.

We define a linear map $L : \Real^p \lrarr \Real^d$: 
\[ L(\vv)_{(U,s)} := \sum_{\omega |_U = s} \vv_{\omega} . \]

\begin{proposition}
\label{inDSprop}
If $L(\vv)$ is induced by an empirical model $e$, then $\vv \in \DS(\Omg)$.
\end{proposition}
\begin{proof}
Since for each $U \in \UU$, $O^U$ is partitioned by the sets $\Omg(U,s)$, we have:
\[ \sum_{\omega \in \Omg} \vv_{\omega} \; = \; \sum_{s \in O^U} e_U(s) \; = \; 1 . \]
\end{proof}

We shall write $\WW$ for the subspace of $\Real^d$ spanned by the set of no-signalling empirical models, and $\VV$ for the subspace spanned by the image of $\DS(E)$ under $L$.
By Proposition~\ref{liimpnsprop}, $\VV \subseteq \WW$. 

We define the set of  \emph{partial contexts}:
\[ \Sigma \, := \, \{ V \subseteq X \mid \exists U \in \UU. \, V \subseteq U \} . \]
We  fix a standard set of outcomes
$O \, := \, \{ 1, \ldots , l \}$, and define $D := \sum_{U \in \Sigma} (l-1)^{\card{U}}$.

\begin{lemma}
$\dim \WW \leq D$.
\end{lemma}
\begin{proof}
For each $U \in \Sigma$ and $p \geq 0$, we define 
\[ \EEj{p}(U) \, := \, \{ s \in O^U \mid \card{s^{-1}(\{ 1 \})} \leq p \} , \]
the set of sections which map at most $p$ measurements to the outcome $1$.
Note that 
\[ \card{\bigcup_{U \in \Sigma} \EEj{0}(U)} = D . \]

Given a section $s$, we write $\sm{j}$ for the section defined by: 
\[ \sm{j}(m) = j, \qquad \sm{j}(m') = s(m'), \;\; (m' \neq m) . \]
Finally, given an empirical model $e$, we define
\[ \ej{p} := \{ e_U(s) \mid U \in \Sigma, s \in \EEj{p}(U) \} . \]

We shall prove that each $e_U(s)$ in $\ej{p+1}$ is a linear combination of values in $\ej{p}$; the fact that every $e_U(s)$ is a linear combination of values in $\ej{0}$ then follows by induction.

Consider $U \in \Sigma$, $s \in \EEj{p}(U)$, and $m \in U$. Let $U' := U \setminus \{ m \}$. Using no-signalling,
\[ e_{U'}(s |_{U'})  \; = \; \sum_j e_U(\sm{j}) . \]
Hence 
\begin{equation}
\label{margineq}
e_U(\sm{1}) \; = \; e_{U'}(s |_{U'}) \, - \,  \sum_{j \neq 1} e_U(\sm{j}) . 
\end{equation}
All the terms on the RHS of this equation are in $\ej{p}$; while every element of $\ej{p+1} \setminus \ej{p}$ can be written in the form of the LHS.

Unwinding the induction, every number $e_U(s)$ is given by a linear combination of values in $\ej{0}$, obtained by back-substitution in (\ref{margineq}). We can regard $\ej{0}$ as a vector $\ww$ in $\Real^D$, where for each $U \in \Sigma$ and $s \in \EEj{0}(U)$, $\ww_{U,s} = e_U(s)$. It follows that there is a linear map $F : \Real^D \lrarr \Real^d$ such that $\WW$ is included in the image of $F$, and hence that $\dim \WW \leq D$.
\end{proof}

\begin{lemma}
$\dim \VV \geq D$.
\end{lemma}
\begin{proof}
Given $U \in \Sigma$ and $s \in \EEj{0}(U)$, we can define the global assignment $\tus : X \rarr O$:
\[ \tus(m) = s(m), \quad (m \in U), \qquad \tus(m) = 1, \quad (m \not\in U) . \]
Note that $\tus = \tusp$ implies $U = U'$ and $s = s'$.
Each such assignment $\tus$ defines a vector $\vus = L(\vv)$, where $\vv_{\tus} = 1$, $\vv_{\omega} = 0$ for $\omega \neq \tus$.
There are clearly  $D = \sum_{U \in \Sigma} (l-1)^{\card{U}}$ such assignments.
We shall show that the set of vectors $\{ \vus \}_{U \in \Sigma, s \in \EEj{0}(U)}$ is linearly independent.

Suppose that $\sum_{U, s} \mus \vus = \mathbf{0}$. 
We shall show that $\mus = 0$ for all $U, s$, by complete induction on $\card{X \setminus U}$. 

Given some $U', s'$, we choose $V \in \UU$ and $s_0 \in O^V$  such that $U' \subseteq V$ and $\tusp |_{V} = s_0$, so that $s_0 |_{U'} = s'$.
Note that, for any $U'', s''$,  $\vusq_{V, s_0} = 1$ if and only if $\tusq |_{V} = s_0$, if and only if $U'' \cap V = U'$, and $s'' |_{U'} = s'$.
If $\tusq \neq \tusp$, we must then have $U'' \supset U'$; so by induction hypothesis, $\musq = 0$.
Using the $(V, s_0)$ component of the vector equation $\sum_{U, s} \mus \vus = \mathbf{0}$, we conclude that $\musp = 0$.
\end{proof}

As an immediate corollary of these two lemmas, we obtain:
\begin{proposition}
\label{VWprop}
$\VV = \WW$.
\end{proposition}

\begin{proposition}
\label{exDSprop}
Let $\ww$ be the vector induced by a no-signalling empirical model $e$. Then for some $\vv \in \DS(\Omg)$, $L(\vv) = \ww$.
\end{proposition}
\begin{proof}
By Proposition~\ref{VWprop}, $\ww = \sum_i \alpha_i L(\vv_i) = L(\sum_i \alpha_i \vv_i)$. Take $\vv := \sum_i \alpha_i \vv_i$. By Proposition~\ref{inDSprop}, $\vv \in \DS(\Omg)$.
\end{proof}

\begin{proposition}
\label{nsprop}
Every empirical model $e$ satisfying No-Signalling is realized by some schv model satisfying Lambda-Independence.
\end{proposition}
\begin{proof}
Given $e$, which induces a vector $\ww$, define $\mO := \vv \in \DS(\Omg)$ with $L(\vv) = \ww$  by Proposition~\ref{exDSprop}. By definition of $L$, we have
\[ e_U(s) \; = \; \sum_{\omega \in \Omg(U, s)} m_{\Omg}(\omega) . \]
We can define a Lambda-Independent schv $m := m_{\UU} m_{\Omg}$, where $m_{\UU}(U) := e(U)$. By Proposition~\ref{liequivprop}, $\hat{m} = e$.
\end{proof}

Combining Propositions~\ref{liimpnsprop} and~\ref{nsprop}, we obtain our main result.

\begin{theorem}
An empirical model is No-Signalling if and only if it is realized by an schv model satisfying Lambda-Independence.
\end{theorem}

\section{General hidden-variable models}

We now consider a more general class of hidden-variable models. In the main cases of interest, these are not more expressive, in terms of the empirical models they realize, than the canonical hidden-variable models we studied in the previous section. They do lead to some additional pathological   behaviour, however, and they lend additional perspective to our discussion.

A (general) hidden-variable model is specified by a set $\Lambda$ of hidden variables, and a measure $h \in \DS(E \times \Lambda)$. We say that $h$ realizes an empirical model $e$ if for all $U \in \UU$, $s \in O^U$:
\[ e(U, s) = \sum_{\lambda \in \Lambda} h(U, s, \lambda) . \]
We can define a measure $\hL \in \DS(\Lambda)$ by marginalization, and for each $\lambda \in \Lambda$, a measure $\hl \in \DS(E)$ as the conditional probability $\hl(U, s) := h(U, s | \lambda)$. We say that $h$ is \emph{Lambda-Independent} if it factors as
\[ h(U, s, \lambda) = \hl(U, s) \hL(\lambda) . \]
We say that $h$ is \emph{Parameter-Independent} \cite{jarrett1984physical,shimony1986events} if for all $U, V \in \UU$, $s \in O^{U \cap V}$, $\lambda \in \Lambda$:
\[ h(s | U, \lambda) = h(s | V, \lambda) . \]

The following result is immediate:

\begin{proposition}
\label{lipinsprop}
If $h$ is Lambda Independent and Parameter Independent, and $h$ realizes $e$, then $e$ is No-Signalling.
\end{proposition}

A canonical hidden-variable model $m$ gives rise to a general hidden-variable model $h_m$ on the hidden-variable set $\Omega$, defined by:
\[ h_m(U, s, \omega) = \left\{ \begin{array}{lr}
m(\omega, U), & \omega |_{U} = s \\
0 & \mbox{otherwise}
\end{array}
\right.
\]

Propositions~\ref{liimpnsprop} and~\ref{lipinsprop} are brought into harmony by the following result.
\begin{proposition}
\label{chvpiprop}
For any canonical hidden variable model $m$, $h_m$ is Parameter-Independent.
\end{proposition}
\begin{proof}
From the definitions, we have that, for $s \in O^{U \cap V}$:
\[ h_m(s | U, \omega) = \left\{ \begin{array}{lr}
1, & \omega |_{U \cap V} = s \\
0 & \mbox{otherwise}
\end{array}
\right.
\]
and similarly for $h_m(s | V, \omega)$.
\end{proof}

Non-canonical hidden-variable models can be less well-behaved. 

\begin{proposition}
There is a hidden-variable model $h$ which is Lambda-Independent but not Parameter-Independent, such that the empirical model $e$ which it realizes is No-Signalling.
\end{proposition}
\begin{proof}
Consider the hidden-variable model $h$ over the cover $\{ \{ a, b_0\}, \{ a, b_1 \}\}$ with outcome set $\{ x_0, x_1, y \}$ and hidden variable set $\{ \lambda_0, \lambda_1 \}$ which assigns probability $1/4$ to each of the following arguments:
\[ (U, s_0, \lambda_0), \quad (V, s_1, \lambda_1),\quad (V, s_2, \lambda_0), \quad (U, s_3, \lambda_1) . \]
Here $U =  \{ a, b_0\}$, $V = \{ a, b_1 \}$, $s_0 = \{ a \mapsto x_0, b_0 \mapsto y \}$,
$s_1 = \{ a \mapsto x_0, b_1 \mapsto y \}$, $s_2 = \{ a \mapsto x_1, b_1 \mapsto y \}$, $s_3 = \{ a \mapsto x_1, b_0 \mapsto y \}$.
Simple calculations show that this model satisfies Lambda-Independence, and the induced empirical model is No-Signalling.
However, $h(x_0 | a, b_0, \lambda_0) = 1/4$, while $h(x_0 | a, b_1, \lambda_0) = 0$, so the model does not satisfy Parameter-Independence.
\end{proof}

\subsection*{Acknowledgements}
Our thanks to Ray Lal, Elliot Lipnowski and Shane Mansfield for valuable comments, and to U.K. EPSRC EP/I03596X/1, U.S. ONR N000141010357, U.S. AFOSR FA9550-12-1-0136, the Templeton Foundation, and NYU Stern School of Business for financial support.

\bibliographystyle{eptcs}

\begin{thebibliography}{1}
\providecommand{\bibitemdeclare}[2]{}
\providecommand{\urlprefix}{Available at }
\providecommand{\url}[1]{\texttt{#1}}
\providecommand{\href}[2]{\texttt{#2}}
\providecommand{\urlalt}[2]{\href{#1}{#2}}
\providecommand{\doi}[1]{doi:\urlalt{http://dx.doi.org/#1}{#1}}
\providecommand{\bibinfo}[2]{#2}

\bibitemdeclare{article}{abramsky2011unified}
\bibitem{abramsky2011unified}
\bibinfo{author}{S.~Abramsky} \& \bibinfo{author}{A.~Brandenburger}
  (\bibinfo{year}{2011}): \emph{\bibinfo{title}{The Sheaf-Theoretic Structure
  Of Non-Locality and Contextuality}}.
\newblock {\sl \bibinfo{journal}{New Journal of Physics}}
  \bibinfo{volume}{13(2011)}, p. \bibinfo{pages}{113036}, \doi{10.1088/1367-2630/13/11/113036}.

\bibitemdeclare{misc}{negprob}
\bibitem{negprob}
\bibinfo{author}{S.~Abramsky} \& \bibinfo{author}{A.~Brandenburger}
  (\bibinfo{year}{2014}): \emph{\bibinfo{title}{An Operational Interpretation
  of Negative Probabilities and No-Signalling Models}}.
\newblock \bibinfo{note}{ArXiv preprint arXiv:1401.2561. Submitted for
  publication}.

\bibitemdeclare{article}{bell1964einstein}
\bibitem{bell1964einstein}
\bibinfo{author}{J.S. Bell} (\bibinfo{year}{1964}): \emph{\bibinfo{title}{{On
  the {E}instein-{P}odolsky-{R}osen paradox}}}.
\newblock {\sl \bibinfo{journal}{Physics}}
  \bibinfo{volume}{1}(\bibinfo{number}{3}), pp. \bibinfo{pages}{195--200}.

\bibitemdeclare{book}{dickson1999quantum}
\bibitem{dickson1999quantum}
\bibinfo{author}{W.M. Dickson} (\bibinfo{year}{1999}):
  \emph{\bibinfo{title}{{Quantum Chance and Non-Locality}}}.
\newblock \bibinfo{publisher}{Cambridge University Press}, \doi{10.1017/cbo9780511524738}.

\bibitemdeclare{article}{Dirac42}
\bibitem{Dirac42}
\bibinfo{author}{P.A.M. Dirac} (\bibinfo{year}{1942}):
  \emph{\bibinfo{title}{The Physical Interpretation of Quantum Mechanics}}.
\newblock {\sl \bibinfo{journal}{Proceedings of the Royal Society of London.
  Series A, Mathematical and Physical Sciences}}
  \bibinfo{volume}{180}(\bibinfo{number}{980}), pp. \bibinfo{pages}{1--40}, \doi{10.1098/rspa.1942.0023}.

\bibitemdeclare{incollection}{feynman1987negative}
\bibitem{feynman1987negative}
\bibinfo{author}{R.P. Feynman} (\bibinfo{year}{1987}):
  \emph{\bibinfo{title}{Negative probability}}.
\newblock In \bibinfo{editor}{B.J. Hiley} \& \bibinfo{editor}{F.D. Peat},
  editors: {\sl \bibinfo{booktitle}{Quantum Implications: Essays in Honour of
  David Bohm}}, \bibinfo{publisher}{Routledge and Kegan Paul}, pp.
  \bibinfo{pages}{235--248}.

\bibitemdeclare{article}{jarrett1984physical}
\bibitem{jarrett1984physical}
\bibinfo{author}{J.P. Jarrett} (\bibinfo{year}{1984}):
  \emph{\bibinfo{title}{{On the physical significance of the locality
  conditions in the Bell arguments}}}.
\newblock {\sl \bibinfo{journal}{No{\^u}s}}
  \bibinfo{volume}{18}(\bibinfo{number}{4}), pp. \bibinfo{pages}{569--589}, \doi{10.2307/2214878}.

\bibitemdeclare{article}{kochen1975problem}
\bibitem{kochen1975problem}
\bibinfo{author}{S.~Kochen} \& \bibinfo{author}{E.P. Specker}
  (\bibinfo{year}{1967}): \emph{\bibinfo{title}{{The problem of hidden
  variables in quantum mechanics}}}.
\newblock {\sl \bibinfo{journal}{Journal of Mathematics and Mechanics}}
  \bibinfo{volume}{17}(\bibinfo{number}{1}), pp. \bibinfo{pages}{59--87}, \doi{10.1007/978-3-0348-9259-9-21}.

\bibitemdeclare{article}{moyal1949quantum}
\bibitem{moyal1949quantum}
\bibinfo{author}{J.E. Moyal} (\bibinfo{year}{1949}):
  \emph{\bibinfo{title}{Quantum mechanics as a statistical theory}}.
\newblock {\sl \bibinfo{journal}{Mathematical Proceedings of the Cambridge
  Philosophical Society}} \bibinfo{volume}{45}(\bibinfo{number}{01}), pp.
  \bibinfo{pages}{99--124}, \doi{10.1017/s0305004100000487}.

\bibitemdeclare{article}{popescu1994quantum}
\bibitem{popescu1994quantum}
\bibinfo{author}{S.~Popescu} \& \bibinfo{author}{D.~Rohrlich}
  (\bibinfo{year}{1994}): \emph{\bibinfo{title}{{Quantum nonlocality as an
  axiom}}}.
\newblock {\sl \bibinfo{journal}{Foundations of Physics}}
  \bibinfo{volume}{24}(\bibinfo{number}{3}), pp. \bibinfo{pages}{379--385}, \doi{10.1007/bf02058098}.

\bibitemdeclare{incollection}{shimony1986events}
\bibitem{shimony1986events}
\bibinfo{author}{A.~Shimony} (\bibinfo{year}{1986}):
  \emph{\bibinfo{title}{{Events and processes in the quantum world}}}.
\newblock In \bibinfo{editor}{R.~Penrose} \& \bibinfo{editor}{C.J. Isham},
  editors: {\sl \bibinfo{booktitle}{{Quantum Concepts in Space and Time}}},
  \bibinfo{publisher}{Oxford University Press}, pp. \bibinfo{pages}{182--203}, \doi{10.1017/cbo9781139172196.011}.

\bibitemdeclare{article}{sudarshan1993new}
\bibitem{sudarshan1993new}
\bibinfo{author}{E.C.G. Sudarshan} \& \bibinfo{author}{T.~Rothman}
  (\bibinfo{year}{1993}): \emph{\bibinfo{title}{A new interpretation of
  {B}ell's inequalities}}.
\newblock {\sl \bibinfo{journal}{International Journal of Theoretical Physics}}
  \bibinfo{volume}{32}(\bibinfo{number}{7}), pp. \bibinfo{pages}{1077--1086}, \doi{10.1007/bf00671790}.

\bibitemdeclare{article}{wigner1932quantum}
\bibitem{wigner1932quantum}
\bibinfo{author}{E.~Wigner} (\bibinfo{year}{1932}): \emph{\bibinfo{title}{On
  the quantum correction for thermodynamic equilibrium}}.
\newblock {\sl \bibinfo{journal}{Physical Review}}
  \bibinfo{volume}{40}(\bibinfo{number}{5}), p. \bibinfo{pages}{749}, \doi{10.1103/physrev.40.749}.

\end{thebibliography}

\end{document}